\documentclass[12pt]{amsart}

\usepackage{amsfonts}
\usepackage{amsmath} 
\usepackage{amssymb}
\usepackage{amsthm}
\usepackage{graphicx}
\usepackage{tikz}

\usetikzlibrary{patterns}
\usetikzlibrary{arrows,shapes,positioning}
\usetikzlibrary{decorations.markings}
\usetikzlibrary[matrix] 
\tikzstyle arrowstyle=[scale=1]

\tikzstyle directed=[postaction={decorate,decoration={markings,
    mark=at position .65 with {\arrow[arrowstyle]{latex}}}}]

\newtheorem{theorem}{Theorem}[section]

\newtheorem{proposition}[theorem]{Proposition}

\newcommand{\lp}{\left(}
\newcommand{\rp}{\right)}

\newcommand{\tH}{\tilde{H}}

\newcommand{\R}{\mathbb{R}}
\newcommand{\N}{\mathbb{N}}
\newcommand{\Nz}{\N_0}
\newcommand{\Z}{\mathbb{Z}}
\newcommand{\Mz}{{M_0}}%\Z_-}}
\newcommand{\C}{\mathbb{C}}
\newcommand{\Q}{\mathbb{Q}}

\newcommand{\cM}{\mathcal{M}}
\newcommand{\cZ}{\mathcal{Z}}

\newcommand{\Klam}{K_\lambda} %{K^{(\lambda)}}

\newcommand{\Dlam}{D_\lambda}%{D^{(\lambda)}}
\newcommand{\Rlam}{R_\lambda}%{R^{(\lambda)}}

\newcommand{\LM}[1]{L_{M,#1}}
\newcommand{\LMn}{\LM{n}}

\newcommand{\hH}{\hat{H}}

\newcommand{\HM}{H_M}
\newcommand{\Hk}{\operatorname{Hk}}
\newcommand{\hk}{\operatorname{hk}}

\newcommand{\Slam}{S_\lambda} %{S^{(\lambda)}}
\newcommand{\Psilam}{\Psi_\lambda}%{\Psi^{(\lambda)}}
\newcommand{\Philam}{\Phi_\lambda}%{\Phi^{(\lambda)}}

\newcommand{\Mlam}{M_{\lambda}}%^{(\lambda)}}

\newcommand{\lspan}{\operatorname{span}}

\makeatletter
\newcommand{\filledbox}{\tikz{\fill circle (3pt)}}
\newcommand{\emptybox}{\tikz{\draw circle (3pt)}}

\newcommand{\Wr}{\operatorname{Wr}}

\newcommand{\bt}{{\boldsymbol{t}}}

\newcommand{\flip}{f}

\newcommand{\bX}{\mathbf{X}}
\newcommand{\bV}{\mathbf{V}}

\begin{document}

%\titlerunning{Coherent States and Rational Extensions}
\title{Extended Coherent  States}
\author{Z.~M.~McIntyre}
\address{Department of Physics, McGill U}
\email{zoe.mcintyre@mail.mcgill.ca}
\author{A. Kasman}
\address{Department of
  Mathematics, College of Charleston}
\email{kasmana@cofc.edu}
\author{ R. Milson}
\address{Department of Mathematics and Statistics, Dalhousie
  University}
\email{rmilson@dal.ca}

\begin{abstract}
  Using the formalism of Maya diagrams and ladder operators,
  we describe the algebra of annihilating operators for the class of
  rational extensions of the harmonic oscillator.  This allows us to
  construct the corresponding coherent state in the sense of Barut and
  Girardello.  The resulting time-dependent function is an exact
  solution of the time-dependent Schr\"odinger equation and a joint
  eigenfunction of the algebra of annihilators.  Using an argument
  based on Schur functions, we also show that the newly exhibited
  coherent states asymptotically minimize position-momentum
  uncertainty.
\end{abstract}

\maketitle

\section{Introduction}

Coherent states are quantum-mechanical states whose dynamics resemble
the behaviour of classical oscillators \cite{GL63,KS85}. The
mathematical treatment of coherent states has opened a number of
interesting new directions in mathematical physics \cite{BG71,PE86}.
In particular, there has been a long standing interest in describing
coherent states related to exactly solvable potentials obtained via
the method of supersymmetric quantum mechanics (SUSYQM)
\cite{FHRO07,BGBM99} Recently, such constructions have been applied to
rational extensions of various exactly solvable potentials
\cite{GDH21} -- the latter are closely related to exceptional
operators and exceptional orthogonal polynomials \cite{GGM13,HHMYZ18}.
The key methodology is showing that these systems have non-trivial
algebras of ladder operators\cite{MQ13}; the coherent states may then
be defined as eigenstates of lowering/annihilator operators.  A recent
article has extended the ladder-operator approach to non-rational
extensions of solvable potentials \cite{CAFMC22}.

In this article, we focus on the supersymmetric partners of the
harmonic oscillator consisting of a modification of the quadratic
potential by a rational function that vanishes at infinity --- hence
the name \emph{rational extension}.  The rational extensions of the
harmonic oscillator are known to possess non-trivial algebras of
ladder operators.  This observation has been successfully exploited in
the study superintegrable systems \cite{MQ13,CP17} and rational
solutions of Painlev\'e equations \cite{BCAF15,GGM21}.

Recent works established the bispectral character of rational
extensions \cite{KM20} and characterized their algebra of ladder
operators using Maya diagrams \cite{GGMM20}.  In the present article,
we combine these two approaches to show that all such rational
extensions admit a natural notion of a coherent state as a joint
eigenvalue of the commutative subalgebra of lowering/annihilator
ladder operators --- we name these objects \emph{extended coherent
  states} (ECS).  The mathematical description of the ECS involves a
certain reduction of the $\tau$-function for the rational solutions of
the KP equation, as is is fully explained in \cite{KM20}, which allows
a convenient description utilizing a technique called a Miwa shift; as
shown below in equation \eqref{eq:Psilamdef}.  As an immediate
corollary we obtain the result that the ECS asymptotically saturate
the Heisenberg position-momentum uncertainty bound as do the the
canonical coherent states (CCS) of the harmonic oscillator
\cite{Sch26} first described by Schr\"odinger.

This article is organized as follows.  Section \ref{sect:prelim}
gathers the necessary background on Maya diagrams and related notions
in combinatorics and integrable systems.  Section \ref{sect:ho}
reviews Hermite polynomials, the harmonic oscillator, the CCS, and
rational extensions.  The key result here is the Miwa-shift formula
\eqref{eq:Psilamdef} for the generating function of the bound states.
Finally, Section \ref{sect:ECS} introduces the ECS as a modification
of the above generating function, defines the annihilator algebra,
establishes the joint eigenvalue properties \eqref{eq:PhilamTM}
\eqref{eq:LqPhial}, and demonstrates the asymptotic minimization of
uncertainty.  The section concludes with an explicit example.
\section{Preliminaries}
\label{sect:prelim}
\subsection{Partitions and Maya diagrams}
A \textit{partition} of a natural number $N\in \Nz$ is a
non-increasing integer sequence $\{ \lambda_i \}_{i\ge 1}$ such that
$ |\lambda| := \sum_i \lambda_i = N$. Implicit in this definition is
the assumption that $\lambda_i=0$ for $i$ sufficiently large.  The
length $\ell$ of $\lambda$ is the number of non-zero elements of the
sequence.  The Young diagram corresponding to $\lambda$ is an
irregular tableaux consisting of $\lambda_i$ cells in rows
$i=1,\ldots, \ell$. Formally,
\[ Y_\lambda = \{ (i,j) \in \N^2\colon 1\le i \le \ell,\; 1\le j \le
  \lambda_i \},\] Note that, unlike the usual convention, we place the
longest row of the Young diagram at the bottom.

The hook
\[ \Hk_\lambda(i,j) = \{ (i,k)\in Y_\lambda \colon j\le k \} \cup \{
  (k,j)\in Y_\lambda \colon i\le k \} \] is the set of cells
connecting cell $(i,j)$ to the rim of the diagram.  The hooklength
$\hk_\lambda(i,j)$ is the cardinality of hook $(i,j)\in
Y_\lambda$. The number
\begin{equation}
  \label{eq:dlamdef}
  d_\lambda = \frac{N!}{\prod_{(i,j)\in Y_\lambda}  \hk_\lambda(i,j) }
\end{equation}
counts the number of standard Young tableaux of shape $\lambda$ and
 corresponds to the dimension of an irreducible representation of the symmetric
group $\mathfrak{S}_N$.

\begin{figure}
  \centering
  \begin{tikzpicture}[scale=0.8]
    \draw[step=1cm,black,very thick] (0,0) grid (2,5);
    \draw[step=1cm,black,very thick] (2,0) grid (4,3);
    \draw[step=1cm,black,very thick] (4,0) grid (5,2);
    \path (0.5,0.5) node {9}
    ++(1,0) node {8}
    ++(1,0) node {5}
    ++(1,0) node {4}
    ++(1,0) node {2};
    \path (0.5,1.5) node {8}
    ++(1,0) node {7}
    ++(1,0) node {4}
    ++(1,0) node {3}
    ++(1,0) node {1};
    \path (0.5,2.5) node {7}
    ++(1,0) node {6}
    ++(1,0) node {2}
    ++(1,0) node {1};
    \path (0.5,3.5) node {3} ++(1,0) node {2};
    \path (0.5,4.5) node {2} ++(1,0) node {1};
  \end{tikzpicture}
  \caption{The Young diagram and corresponding hooklengths for the
    partition $(5,5,4,2,2)$.}
\end{figure}

Closely related to partitions is a concept called a \emph{Maya
  diagram}.  We say that $M\subset \Z$ is a Maya diagram if $M$
contains a finite number of positive integers and excludes a finite
number of negative integers.
% \[ K_M^+ = \{ m\in M \colon m\geq 0\},\quad K_M^- =\{ m\in \Z\setminus
%   M \colon m<0 \} \] are finite sets.  In other words, a Maya diagram
% is a subset of $\Z$ that
Let $\cM$ denote the set of all Maya diagrams.  For $M\in\cM$ and
$n\in\Z$, $M+n=\{m+n:m\in M\}$ is also a Maya diagram.  Thus, $\cM$
admits a natural $\Z$ action by translations.

We will refer to the equivalence class $M/\Z$ as an \emph{unlabelled}
Maya diagram.  Intuitively, an unlabelled Maya is a horizontal
sequence of filled $\filledbox$ and empty $\emptybox$ states beginning
with an infinite segment of $\filledbox$ and terminating with an
infinite segment of $\emptybox$.  Here, $\filledbox$ in position $m$
is taken to indicate membership $m\in M$.  A choice of origin serves
to convert an unlabelled Maya diagram to a subset of $\Z$.  The index
of a Maya diagram $M\in \cM$ is the integer
\[ \sigma_M:= \#\{ m\in M \colon m\geq 0\}- \#\{ m\notin M \colon m<0
  \};\] i.e., the difference between the number of $\filledbox$ to
the right of the origin and the number of $\emptybox$ to the left of
the origin.  Evidently, $\sigma_{M+n} = \sigma_M+n$.

There is a natural bijection between the set of partitions and the set
of unlabelled Maya diagrams. To visualize this bijection,
represent a filled state with a unit downward arrow and an empty
state with a unit right arrow.  As can be seen in Figure
\ref{fig:bentmaya}, the resulting ``bent'' Maya diagram traces out the
boundary of the Young diagram of the corresponding partition
$\lambda$; see \cite{GGM18} for more details.

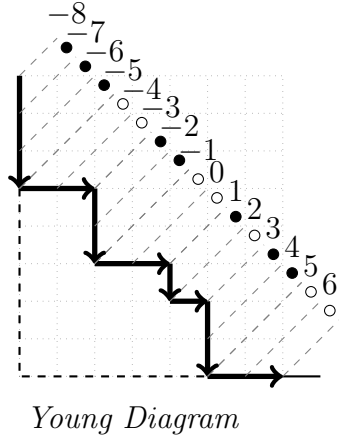
\begin{figure}
  \centering
  \begin{tikzpicture}[scale=0.5]

\draw[step=1cm,gray,dotted,very thin] (0,0) grid (7,8);

\draw[dashed,thick] (0,5) -- +(0,-5) -- +(5,-5);

\draw[line width=2pt,->]    (0,8)  -- ++(0,-1)-- ++(0,-1)-- ++(0,-1);
\draw[line width=2pt,->]    (0,8-3)  -- ++(1,0) -- ++(1,0);
\draw[line width=2pt,->]    (0+2,8-3)  -- ++(0,-1)-- ++(0,-1);
\draw[line width=2pt,->]    (0+2,8-5)  -- ++(1,0) -- ++(1,0);
\draw[line width=2pt,->]    (0+4,8-5)  -- ++(0,-1);
\draw[line width=2pt,->]    (0+4,8-6)  -- ++(1,0 );
\draw[line width=2pt,->]    (5,2)  -- ++(0,-1)-- ++(0,-1);
\draw[line width=2pt,->]    (5,0)  -- ++(2,0);

\draw[thick] (5,0) -- ++(3,0);

\foreach \y  in {5,...,8} {
\draw[dashed,color=gray] (0,\y) -- (5-\y/2,5+\y/2);
};  

\foreach \x  in {0,...,2} {
\draw[dashed,color=gray] (\x,5) -- (2.5+\x/2,7.5-\x/2);
};  

\foreach \y  in {3,...,5} {
\draw[dashed,color=gray] (2,\y) -- (6-\y/2,4+\y/2);
};  

\foreach \x  in {2,...,4} {
\draw[dashed,color=gray] (\x,3) -- (3.5+\x/2,6.5-\x/2);
};  

\foreach \y  in {2,...,3} {
\draw[dashed,color=gray] (4,\y) -- (7-\y/2,3+\y/2);
};  

\foreach \y  in {0,...,2} {
\draw[dashed,color=gray] (5,\y) -- (7.5-\y/2,2.5+\y/2);
};  

\foreach \x  in {5,...,7} {
\draw[dashed,color=gray] (\x,0) -- (5+\x/2,5-\x/2);
};

\draw[fill=black]
    (1.25,8.75) circle (4pt) node[above=4pt] {$-8$}
  ++ (0.5,-0.5)  circle (4pt) node[above=4pt] {$-7$}
  ++ (0.5,-0.5)  circle (4pt) node[above=4pt] {$-6$};

\draw  (1.25,8.75) ++ (1.5,-1.5)
  circle (4pt) node[above=4pt] {$-5$}
 ++ (0.5,-0.5)  circle (4pt) node[above=4pt] {$-4$};

\draw[fill=black]
    (1.25+2.5,8.75-2.5) circle (4pt) node[above=4pt] {$-3$}
  ++ (0.5,-0.5)  circle (4pt) node[above=4pt] {$-2$};

\draw  (1.25+3.5,8.75-3.5)
  circle (4pt) node[above=2pt] {$-1$}
 ++ (0.5,-0.5)  circle (4pt) node[above=2pt] {$0$};

\draw[fill=black]  (1.25+4.5,8.75-4.5)
  circle (4pt) node[above=2pt] {$1$};

\draw  (1.25+5,8.75-5)
  circle (4pt) node[above=2pt] {$2$};

\draw[fill=black]  (1.25+5.5,8.75-5.5)
  circle (4pt) node[above=2pt] {$3$}
++ (0.5,-0.5)  circle (4pt) node[above=2pt] {$4$};

\draw  (1.25+6.5,8.75-6.5)
  circle (4pt) node[above=2pt] {$5$}
 ++ (0.5,-0.5)  circle (4pt) node[above=2pt] {$6$};

\path (3,-0.5) node[font=\itshape,anchor=north,align=center]
{Young Diagram};

% \draw[thick,->]  (3.25,9.75) -- node[above,rotate=-45,font=\itshape] {Maya Diagram} +(6,-6);

\end{tikzpicture}

\caption{The bent Maya diagram with index set $K=\{-5,-4,-1,1,3,4 \}$
  is the rim of the Young diagram of the corresponding partition
  $\lambda=(5,5,4,2,2)$.}
\label{fig:bentmaya}
\end{figure}

\noindent
The bijection may also be described as $\lambda\to \Mlam/\Z$ where $\lambda$ is a  partition and
% define the strictly decreasing sequence
% \begin{equation}
%   \label{eq:midef} m_i = \lambda_i - i,\quad i=1,2,\ldots.
% \end{equation}
% The set
\begin{equation}
  \label{eq:Mlam}
  \Mlam = \{ \lambda_i-i \}_{i\in \N}
\end{equation}
The bijection claim is justified by showing that $\sigma_{\Mlam}=0$
and that every equivalence class in $\cM/\Z$ contains a unique Maya
diagram with zero index.
% because $m_{i+1} = m_i-1$ for $i\geq {\ell}+1$.
% \begin{proof}
% Indeed, let $M\subset \Z$ be a Maya
% diagram and $m_1>m_2>\dotsm$ the strictly decreasing sequential
% ordering of the elements $m\in M$.  Set
% \begin{equation}
%   \label{eq:sigmaM}
%   \sigma_M = \# \KMpos- \# \KMneg\\
% \end{equation}
% where $\#$ denotes the cardinality of a finite set. By construction,
% \begin{equation}
% \label{indextranslate}
% \sigma_{M+n}=\sigma_M+n,\quad n\in \Z.
% \end{equation}
% Hence, there exists a sufficiently large $\ell$ such that
% $m_i=-i+\sigma_M$ for all $i> \ell$. 
% Next, let
% \begin{equation}
%   \label{eq:lamfromM}
%   \lambda_i = m_i+ i-\sigma_M,\quad i=1,2,\ldots.
% \end{equation}
% It follows that $\lambda_i=0$ for all $i> \ell$;
% i.e., $\lambda$ is a partition. Furthermore, by construction,
% \[ M = \Mlam + \sigma_M.\]
% \end{proof}

The flip $\flip_{k}$ at position $k\in\mathbb{Z}$ is the involution
$\flip_k:\mathcal{M}\rightarrow\mathcal{M}$ defined by
\begin{equation}
\flip_k:M\mapsto \begin{cases}
M\cup\{k\},& k\notin M\\
M\setminus\{k\},&k\in M
\end{cases}.
\end{equation}
\noindent
In the event that $k\notin M$, the flip $\flip_k$ is said to act on $M$
by a state-deleting transformation $\emptybox\rightarrow\filledbox$,
while in the opposite scenario ($k\in M$), it is said to act by a
state-adding transformation $\filledbox\rightarrow\emptybox$.

Let $\cZ$ denote the set of all finite subsets of $\Z$. For a finite
set of integers $K=\{ k_1,\ldots, k_p \}\in \cZ$ we define the
corresponding multi-flip to be the transformation
$\flip_K:\cM\rightarrow\cM$ defined according to
\begin{equation}
\flip_K(M)=(\flip_{k_1}\circ\dotsm\circ \flip_{k_p})(M).
\end{equation}

Observe that multi-flips are also involutions. This means that Maya
diagrams together with multifips have the natural structure of a
complete graph $(\cM,\cZ)$.  The edge connecting $M_1, M_2\in \cM$ is
the unique multi-flip $\flip_K$ such that $\flip_K(M_1)=M_2$
and $\flip_K(M_2)=M_1$. The corresponding $K\in \cZ$ is given as
the symmetric set difference
\begin{equation}
\label{edge}
K=M_1\ominus M_2=M_2\ominus M_1
\end{equation}
% where
% \[M_1 \ominus M_2 := (M_1\setminus M_2) \cup (M_2\setminus M_1) \] is
% the symmetric difference operation.

Since $(\cM,\cZ)$ is a complete graph, we can define a bijection
$\cZ\rightarrow \mathcal{M}$ given by $K\mapsto \flip_K(\Mz)$, where
$  \Mz = \{ m \in \Z \colon m<0 \}$
denotes the trivial Maya diagram, and where $K=M\ominus M_0$ is the \emph{index set} of the Maya diagram $M$.

The hooklength formula \eqref{eq:dlamdef} can be re-expressed in terms
of a Maya diagrams and index sets as follows.  Let $\lambda$ be a
partition of length $\ell$. Define $\Klam$ to be the index set of
$\Mlam+\ell$. Then, $k_i= \lambda_i-i+\ell,\; i=1,2,\ldots,\ell$ is
the decreasing enumeration of $\Klam$, and
\[\frac{N!}{d_\lambda}= \prod_{i,j} \hk_\lambda(i,j) = \frac{\prod_i
    k_i!}{\prod_{i<j}(k_i-k_j)}.\]

\subsection{Vertex operators and Schur functions}
For $k\in \Nz$, define the ordinary Bell polynomials
$B_k(t_1,\ldots, t_k) \in \Q[t_1,\ldots, t_k]$ as the
coefficients of the power generating function
\begin{equation}
  \label{eq:Bgf}
    \exp\left(\sum_{k=1}^\infty t_k z^k\right)
    = \sum_{k=0}^\infty B_k(t_1,\ldots, t_k) z^k,
\end{equation}
where $ \bt=(t_1,t_2,\ldots)$.  The multinomial formula implies that
\begin{equation}
  \label{eq:Bksum}
  \begin{aligned} B_k(t_1,\ldots, t_k) &= \sum_{ \Vert \mu\Vert=k} \frac{t^{\mu_1}_1}{\mu_1!}
\frac{t^{\mu_2}_2}{\mu_2!}  \cdots \frac{t^{\mu_{k}}_{k}}{\mu_{k}!},\qquad
\Vert \mu \Vert = \mu_1 + 2\mu_2 + \cdots + {k} \mu_{k}\\ &=
\frac{t_1^k}{k!} + \frac{t_1^{k-2}t_2}{(k-2)!} + \cdots + t_{k-1} t_1
+ t_k.
  \end{aligned}
\end{equation}
For a partition $\lambda$ of $N$, define the Schur function
$\Slam(t_1,\ldots, t_N)\in \Q[t_1,\ldots, t_N]$ to be the multivariate
polynomial
 \begin{equation}
   \label{eq:Slamdef} \Slam =
   \det(B_{m_i+j})_{i,j=1}^\ell, \quad m_i = \lambda_i-i,%{{\ell(\lambda)}}, 
 \end{equation}
 % where
 % \[ m_i = \lambda_i-i, \]
 where $B_k=0$ when $k<0$.  Moreover, since
 \[ \partial_{t_i} B_j(t_1,\ldots, t_j) = B_{j-i}(t_1,\ldots,
   t_{j-i}),\quad j\geq i,\] we may re-express \eqref{eq:Slamdef} in
 terms of a Wronskian determinant,
\begin{equation}
  \label{eq:Slamwronsk}
  \Slam = \Wr[B_{m_\ell+\ell},\ldots, B_{m_1+\ell}],
\end{equation}
where the Wronskian is taken with respect to $t_1$.
% and where
% \begin{align}
%   \label{eq:cKlamn}
%   k_i = \lambda_i-i + \ell,\quad i=1,\ldots, \ell.
% \end{align}

Let $\bX_m= \bX_m(\bt,\partial_\bt),\; m\in \Z$, be the operators
defined by the generating function
\begin{equation}
  \label{eq:bVdef}
\begin{aligned}
  \bV(\bt,\partial_\bt,z)
  &=\exp\lp \sum_{k=1}^\infty  t_k z^k \rp \exp\lp \sum_{j=1}^\infty
    - j^{-1}\partial_{t_j} z^{-j} \rp \\
  &= \sum_{m=-\infty}^\infty \bX_m(\bt,\partial_\bt)    z^m.
  % \bV^- =\exp\lp \sum_{k=1}^\infty -t_k z^k \rp \exp\lp \sum_{j=1}^\infty
  % j^{-1}\partial_{t_j} z^{-j} \rp = \sum_{m=-\infty}^\infty \bX^*_m
  % z^m, 
\end{aligned}
\end{equation}
% In principle, each $X_m$ is an infinite linear combination
% of $\partial_{t_1},\partial_{t_2}, \ldots$ with a well-defined
  % action on a polynomial in $t_1,t_2,\ldots$.
Expanding the above formulas gives
\begin{align}
  \label{eq:Xm+}
  \bX_m &= \sum_{j=0}^\infty B_{j+m}( t_1,\ldots,  t_k) B_j\lp
          \partial_{t_1},\ldots,  j^{-1}
          \partial_{t_j}\rp,\; m\ge 0;\\  
  \label{eq:Xm-}
  \bX_m &= \sum_{j=0}^\infty B_{j}( t_1,\ldots,  t_k) B_{j-m}\lp
          -\partial_{t_1},\ldots, - j^{-1}          \partial_{t_j}\rp,\; m<0.
  % \label{eq:X*m+}
  % \bX^*_m &= \sum_{j=0}^\infty B_{j+m}(-t_1,\ldots, -t_k) B_j\lp
  %           \partial_{t_1},\ldots, j^{-1} \partial_{t_j}\rp,\; m\ge 0;\\
  % \label{eq:X*m-}
  % \bX^*_m &= \sum_{j=0}^\infty B_{j}(-t_1,\ldots, -t_k) B_{j-m}\lp
  %           \partial_{t_1},\ldots, j^{-1} \partial_{t_j}\rp,\; m<0.
\end{align}
It can be shown that the above operators obey the 
fundamental relation
\begin{align}
  \label{eq:XmXn}
  &\bX_m \bX_n + \bX_{n-1} \bX_{m+1} = 0.
%  &\bX^{\pm}_m \bX^{\mp}_n + \bX^{\mp}_{n+1} \bX^{\pm}_{m-1} = \delta_{m+n,0}.
\end{align}
Despite the fact that the $\bX_m(\bt,\partial_\bt)$ are differential
operators involving infinitely many variables, they have a
well-defined action on polynomials.  In particular, when applied to
Schur functions, they function as multi-variable raising operators.
\begin{proposition}
  \label{prop:SlamXl}
  For every partition $\lambda$ of length $\ell$, we have
  \begin{equation}
    \label{eq:SlamXlam}
    \Slam = \bX_{\lambda_1} \cdots \bX_{\lambda_{\ell}} 1,
  \end{equation}
  where $1$ is the Schur function corresponding to the trivial partition.
\end{proposition}
\noindent
The proof of \eqref{eq:XmXn}--\eqref{eq:SlamXlam} can be found in
\cite[Appendix A]{NY99}. As an immediate corollary we obtain the
following two results \cite{KM20}.
\begin{proposition}\label{prop:XmSlam}
  Let $\lambda$ be a partition, $\Mlam$ the Maya diagram as per
  \eqref{eq:Mlam}.
  For $m\notin \Mlam$ let $m\triangleright\lambda$ denote the
  partition
  \begin{equation}
    \label{eq:lmlist}
    \lambda_1 -1, \ldots , \lambda_j -1, m+j, \lambda_{j+1},
    \lambda_{j+2},\ldots,
  \end{equation}
  where $j$ is the smallest natural number such that
  $m+j \geq \lambda_{j+1}$. Then,
  % , and
  % \begin{equation}
  %   \label{eq:Jlam}
  %   \Jlam = \Z\setminus \Mlam
  % \end{equation}
  % the corresponding integer complement.
%  Then, for every $m\in\Z$ we have
  \begin{equation}
    \label{eq:Xmact}
    \bX_m  \Slam =
    \begin{cases} 
      (-1)^{\# \{ k\in \Mlam : k> m \}} S_{m\triangleright\lambda} &
      \text{ if } m\notin \Mlam\\
      0 & \text{ if } m\in \Mlam
    \end{cases}.
  \end{equation}
\end{proposition}
\noindent
By construction, the action of $\bV(\bt,z)$ on a polynomial
$P(\bt)\in \C[t_1,\ldots, t_n]$ is
{\small
\begin{equation}
  \label{eq:bVP}
  \bV(\bt,z) P(\bt) = \exp\lp \sum_{k=1}^\infty t_k z^k \rp P\!\lp t_1-
  z^{-1}, t_2 - \frac{z^{-2}}{2}, \ldots, t_n - \frac{z^{-n}}{n}\rp.
\end{equation}}
\noindent
Proposition~\ref{prop:XmSlam} allows the action of $\bV(\bt,z)$ on a Schur
polynomial to be conveniently written in terms of the ``insertion''
procedure $m\triangleright\lambda$:
\begin{theorem}
  Let $\lambda$ be a partition. With the above notation, we have
  \begin{equation}
    \label{eq:Vlam}
    \bV(\bt,z) \Slam(\bt) = \sum_{m\notin \Mlam}
    (-1)^{\# \{ k\in \Mlam : k> m \}} S_{m\triangleright\lambda}(\bt)  z^m.
  \end{equation}
\end{theorem}
\section{The harmonic oscillator and its rational extensions}
\label{sect:ho}
\subsection{Hermite polynomials}
Hermite polynomials $\{H_n(x)\}_{n\in \Nz}$ are classical orthogonal
polynomials that satisfy the second-order eigenvalue equation
\begin{equation}
  \label{eq:hermde}
  y''-2xy' = 2n y,\quad y= H_n(x),
\end{equation}
and the orthogonality relation
\begin{equation}
  \label{eq:hortho}
  \int_{\R} H_m(x) H_n(x) e^{-x^2} dx = \sqrt{\pi}\, 2^n
  n! \delta_{n,m}.
\end{equation}
The above is equivalent  to the 3-term recurrence relation
\begin{equation}
  \label{eq:h3term}
  H_{n+1}(x) =  2x H_n(x) - 2n H_{n-1}(x),\quad H_0(x)=1
\end{equation}
% $\{H_n(x)\}_{n\in \Nz}$,\; n=0,1,\ldots$, are univariate
% polynomials defined by
The generating function for the Hermite polynomials is
\begin{align}
  \label{eq:hermgf}
    e^{xz - \tfrac14 z^2}
    &= \sum_{n=0}^\infty H_n(x) \frac{z^n}{2^nn!},
    % \label{eq:hermgfc}
    % e^{xz +\tfrac14 z^2}
    % &= \sum_{n=0}^\infty \tH_n(x) \frac{z^n}{2^nn!},
\end{align}
which can be readily established by observing that
\begin{equation}
  \label{eq:Dx2Dz}
  (\partial_x+ 2\partial_z) \lp e^{xz - \tfrac14z^2- x^2}\rp =
  (\partial_x+ 2\partial_z) e^{-(x-z/2)^2} =  0,
\end{equation}
and by applying the well-known Rodrigues formula
\begin{equation}
  \label{eq:HnRodrigues}
  H_n(x)=(-1)^ne^{x^2}\frac{d^n}{dx^n}e^{-x^2},\;  n\in \Nz
\end{equation}

% To prove this statement \eqref{eq:hermgf}, write
% \[ e^{xz - \tfrac14 z^2} = \sum_{n=0}^\infty f_n(x)
%   \frac{z^n}{2^nn!}\] and observe that
% Hence,
% \[
%   \partial_x e^{xz - \tfrac14z^2-x^2}
%   = -2\partial_z e^{xz - \tfrac14z^2-x^2}
%   = - \sum_{n=0}^\infty f_{n+1}(x) e^{-x^2} \frac{z^n}{2^n n!}.
% \]
% It follows that
% \[ \frac{d}{dx} \lp f_n(x) e^{-x^2}\rp = - f_{n+1}(x)e^{-x^2},\] so since
% $f_0(x)=1$, we conclude that $f_n(x) = H_n(x)$ for all
% $n=0,1,\dots$.  

Comparison of \eqref{eq:hermgf} with \eqref{eq:Bgf} shows that the
Hermite polynomials are specializations of Bell polynomials:
\[
  \begin{aligned}
    H_n(x) &= n! 2^n B_n(x,-\tfrac14,0,\ldots)   .\\
  \end{aligned}
\]
Applying \eqref{eq:Bksum} then gives the well-known formula
\[ H_n(x) = \sum_{j=0}^{\lfloor n/2\rfloor} (-1)^j \frac{n!}{(n-2j)!
    j!} (2x)^{n-2j}.\] 

In the sequel, we will also make use of the conjugate Hermite
polynomials:
\begin{equation}
  \label{eq:tHndef}
  \tH_n(x) = n! 2^n B_n(x,\tfrac14,0,\ldots) =i^{-n} H_n(i x),\quad n\in \Nz
\end{equation}

\subsection{The canonical Hamiltonian and coherent state}

Write $p= i \partial_x$, so that
\[ T(x,\partial_x) = p^2+x^2=- \partial_x^2 + x^2 \] is the
Hamiltonian of the quantum harmonic oscillator.  We say that a
function $\psi(z)$ is quasi-rational if its log-derivative,
$\psi'(z)/\psi(z)$, is a rational function.  The quasi-rational
eigenfunctions of $T$ are the Hermite functions
\begin{equation}
  \label{eq:psindef}
\psi_n(x)=\begin{cases}
e^{-\tfrac{x^2}{2}}H_n(x), & n\geq 0\\
e^{\tfrac{x^2}{2}}\tH_{-n-1}(x), & n<0.
\end{cases},
\end{equation}
We now show that the $\psi_n,\; n\ge 0$, represent the bound states of
the harmonic oscillator, while the $\psi_n,n<0$
% do not satisfy the
% boundary conditions at $\pm\infty$ and instead
represent virtual states.
% where $H_n(x)$ and  $\tH_n(x)$ are the Hermite and conjugate Hermite
% polynomials defined in \eqref{eq:HnRodrigues} and \eqref{eq:tHndef}. 
% It will be instructive to prove \eqref{eq:Tpsin} using generating
% functions.
Multiplication of \eqref{eq:hermgf} by $e^{-x^2/2}$ yields
the generating function for the bound states:
\begin{equation}
  \label{eq:psingf}
  \Psi_0(x,z) := e^{-\tfrac12(x-z)^2+\tfrac14z^2}= \sum_{n=0}^\infty \psi_n(x) \frac{z^n}{2^n n!}.
\end{equation}
By a direct calculation, we have
\begin{equation}
  \label{eq:TPsi0}
   T(x,\partial_x) \Psi_0(x,z) = (2 z \partial_z+1)\Psi_0(x,z). 
\end{equation}
Applying the above relation to \eqref{eq:psingf} and comparing the
coefficients of the resulting power series then returns the desired
eigenvalue relation
\begin{equation}
  \label{eq:Tpsin}
  T\psi_n = (2n+1) \psi_n,\quad n\in \Z.
\end{equation}

The classical ladder operators
\begin{equation}
  \label{eq:L+idef}
  \begin{aligned}
    L_{\mp}(x,\partial_x) := \partial_x \pm x
  \end{aligned}
\end{equation}
satisfy the intertwining relations
\[ T L_- = L_- (T-2),\qquad T L_+ = L_+(T+2).\]
An immediate consequence are the lowering and raising relations:
\begin{equation}
  \label{eq:lrrel}
    L_- \psi_n =
    \begin{cases}
      2n \psi_{n-1}, &n \ge 0\\
      \psi_{n-1}, & n<0      
    \end{cases}\qquad
    L_+ \psi_n =
    \begin{cases}
      \psi_{n+1},& n> 0\\
      2(n+1) \psi_{n+1} & n\le 0      
    \end{cases}
\end{equation}
% Relations \eqref{eq:lrrel} can also be established using generating
% functions; it suffices to observe that
% \begin{align}
%   \label{eq:L-Psi0}
%   L_-(x,\partial_x) \Psi_0(x,z)
%   &= z \Psi_0(x,z) = \sum_{n=0}^\infty 2n \psi_{n-1}(x)
%     \frac{z^n}{2^n n!};\\
%   L_+(x,\partial_x) \Psi_0(x,z) &= 2\partial_z \Psi_0(x,z) = \sum_{n=0}^\infty \psi_{n+1}(x)
%                                   \frac{z^n}{2^n n!}.
% \end{align}

Now define  the canonical coherent state (CCS) to be
\begin{equation}
  \label{eq:Phi0def}
  \Phi_0(x,t;\alpha) :=
%  \exp\lp -it + \tfrac12 x^2 - (x-\tfrac12 e^{-2   it}\alpha)^2\rp=
 e^{-it} \Psi_0(x,\alpha e^{-2it}).
\end{equation}
The change of variable $z=\alpha e^{-2it}$ transforms \eqref{eq:Dx2Dz}  and
 \eqref{eq:TPsi0} into
\begin{equation}
  \begin{aligned}
  L_- \Phi_0 &= \alpha e^{-2it}  \Phi_0\\
  T \Phi_0 &= i \partial_t\Phi_0.
  \end{aligned}
\end{equation}
Thus, the CCS is an eigenfunction of $L_{-}$ and an exact solution to
the time-dependent Schr\"odinger equation.
% Hence, the canonical coherent state may also be regarded a generating
% function for the bound states of the harmonic oscillator.  Indeed, the
% change of variables \eqref{eq:Phi0def} transforms the eigenvalue
% relation \eqref{eq:L-phi} into \eqref{eq:L-Psi0}, and the
% TDSE \eqref{eq:IDtTphi} into relation \eqref{eq:TPsi0}.

\subsection{Hermite Pseudo-Wronskians}
Let $M\in \cM$ be a Maya diagram and $K=\{k_1,\ldots, k_p\}$ the
corresponding index set arranged in increasing order
$k_1<\cdots k_q<0 \le k_{q+1} < \cdots k_p$. Define the
pseudo-Wronskian
\begin{equation}\label{eq:pWdef}
  \HM =
  \det \begin{vmatrix}
    \tH_{-k_1-1} & \tH_{-k_1} & \ldots & \tH_{-k_1+p-1}\\
    \vdots & \vdots & \ddots & \vdots\\
    \tH_{-k_q-1} & \tH_{-k_q} & \ldots & \tH_{-k_q+p-1}\\
    H_{k_{q+1}} & D_x H_{k_{q+1}} & \ldots & D_x^{p-1}H_{k_{q+1}}\\
    \vdots & \vdots & \ddots & \vdots\\
    H_{k_p} & D_x H_{k_p} & \ldots & D_x^{p-1}H_{k_p}
  \end{vmatrix}.
\end{equation}
% \begin{equation}\label{eq:pWdef1}
%   \HM = e^{-qx^2}\Wr_x[ e^{x^2} \tH_{-k_1-1},\ldots, e^{x^2}
%   \tH_{-k_q-1}, H_{k_{q+1}},\ldots H_{k_p} ],
% \end{equation}
% where $q$ is the cardinality of the negative elements of $K$, and
% where $\Wr$ denotes the usual Wronskian determinant.   The same article
% also showed that
One can show \cite{GGM18} that the normalized polynomial
\begin{equation}
  \label{eq:hHdef}
\hH_M=\frac{(-1)^{(p-q)q}H_M}{\prod_{i<j\le q}2(k_j-k_i)\prod_{q+1\le i<j}2(k_j-k_i)}
\end{equation}
is translation-invariant
\begin{equation}
  \label{eq:HM+n}
  \hH_M=\hH_{M+n},\quad n\in\Z.
\end{equation}
and hence may be regarded as a function of the corresponding partition
$\lambda$.  Moreover, in \cite{KM20}, it was shown that the normalized
Hermite pseudo-Wronskian \eqref{eq:hHdef} has the following expression
in terms of Schur functions:
\begin{equation}
  \label{eq:HMSlam}
  \hH_M(x) =  \frac{2^{N}N!}{d_\lambda} \, \Slam(x,-\tfrac14,0,\ldots ),
\end{equation}
where $N=|\lambda|$ and where $d_\lambda$ is the combinatorial factor
defined in \eqref{eq:dlamdef}.

\subsection{ Rational Extensions of the Harmonic Oscillator}
Let $M\in \cM$ be a Maya diagram. The
pseudo-Wronskian defined in \eqref{eq:pWdef} can now be expressed  \cite{GGM18}
simply as
\begin{equation}
  \label{eq:HM}
  H_M(x)=e^{\sigma_M\tfrac{x^2}{2}}\Wr[\psi_{k_1}(x),\dots,\psi_{k_p}(x)],
\end{equation}
where $\psi_n(x),\; n\in \Z$, are the quasi-rational eigenfunctions
\eqref{eq:psindef}, and where $\sigma_M$ is the index of $M$.  The
potential
\begin{align}
\label{rationalext}
  U_M(x)
  &=x^2-2\frac{d^2}{dx^2}\log\Wr[\psi_{k_1},\dots,\psi_{k_p}]\\ \nonumber
  &=x^2+2\left(\frac{H'_M}{H_M}\right)^2-\frac{2H_M^{''}}{H_M}-2\sigma_M   
\end{align}
is a rational extension of the harmonic oscillator potential, so
called because the terms following the $x^2$ in~\eqref{rationalext}
are all rational.  The corresponding Hamiltonian 
\begin{equation}
  T_M:=-\frac{d^2}{dx^2}+U_M
\end{equation}
is  exactly solvable \cite{GGM13} with eigenfunctions
\begin{align}
\label{eigenstates}
  \psi_{M,m}=
  e^{\epsilon_M(m)\tfrac{x^2}{2}}\frac{\hH_{M,m}}{\hH_M},\quad
  \epsilon_M(m) =
  \begin{cases}
    -1 & \text{ if } m\notin M\\
    +1 & \text{ if } m\in M\\
  \end{cases},\quad m\in \Z
\end{align}
where $(M,m):=\flip_m(M)$. The eigenvalue relation is
\begin{equation}
  \label{eq:TMpsiMn}
  T_M\psi_{M,m}=(2m+1)\psi_{M,m},\quad m\in \Z.
\end{equation}
The numerators $\hH_{M,m}(x),\; m\notin M$ are known as
exceptional Hermite polynomials \cite{GGM13}.  Relation
\eqref{eq:HM+n} implies that $T_M$ and the corresponding
eigenfunctions are translation covariant:
\begin{equation}
    T_{M+n}=T_M+2n,\quad     \psi_{M+n,m+n}= \psi_{M,n}
,\; n\in\Z
\end{equation}
Thus, the unlabelled Maya diagram is a representation of the spectrum, with
\begin{equation}
  \label{eq:IMdef}
  I_M := \Z\setminus M = (\Z\setminus \Mlam) + \sigma_M,
\end{equation}
serving as the index set for the bound states (the ones with label
$\emptybox$).

As regards regularity, it should be noted that by the Krein-Adler
theorem \cite{GGM13}, $H_M$ has no real zeros if and only if all
finite $\filledbox$ segments of $M$ have even size. It is precisely
for such $M$ that $T_M$ corresponds to a self-adjoint operator and
that the eigenfunctions $\psi_{M,m},\; m\in I_M$ are
square-integrable. If this condition fails, then one still has
orthogonality and self-adjointness if one deforms the contour of
integration away from the singularities \cite{HHV16}.  However, in the
presence of singularities in $U_M$, the resulting inner product is no
longer positive-definite, but rather has a finite signature.

The generating function for the bound states of a rational extension
can be given using the Miwa shift formula  \eqref{eq:bVP}.
\begin{proposition}
For a partition $\lambda$, define
\begin{equation}
  \label{eq:Psilamdef}
  \Psilam(x,z) = \frac{\Slam\lp x-z^{-1}, -\tfrac14
    -\tfrac12z^{-2},-\tfrac13 z^{-3},\ldots\rp}{\Slam(x,-\tfrac14,0,\ldots)} \Psi_0(x,z).
\end{equation}
Let $M\in \cM$ be a Maya diagram and $\lambda$ the corresponding
partition. Then,
\begin{equation}
  \label{eq:Psilamgf}
  \Psilam(x,z) = \sum_{m\in I_M}  \psi_{M,m}(x)
  \frac{\prod_{i=1}^\ell (m-m_i)}{(m-\sigma_M+\ell)!}
  \lp\frac{z}{2}\rp^{m-\sigma_M}, 
\end{equation}
where $m_1>m_2>\cdots$ is the decreasing enumeration of $M$.
\end{proposition}
\begin{proof}
  This follows from \eqref{eq:bVdef}, \eqref{eq:Xmact} and \eqref{eq:bVP}.
\end{proof}

\noindent
Observe that if $M=\Mz$ is the trivial Maya diagram, then
\eqref{eq:Psilamgf} reduces to the classical generating function shown
in \eqref{eq:psingf}.

%\vspace{2ex}
% \comment{Terminology: intertwining operator = ladder operator.
%   Annihilation  operators are the ones that have bound states in the
%   kernel.  Creation operators are adjoints of annihilation operators.
%   They are also intertwining operators, but with a negative index.}

\section{Extended coherent states}
\label{sect:ECS}
\subsection{Ladder Operators}
% In this section, we introduce ladder operators for the rational
% extensions $T_M,\; M\in \cM$, defined above.  
% Intertwining relations
% have their origins in supersymmetric quantum mechanics (SUSYQM).  For
% differential operators $A$, $T_1$, $T_2$, we say that $A$ intertwines
% $T_1$ and $T_2$ if
Let $T,A$ be differential operator.  We say that $A$ is a ladder
operator for $T$ if 
\begin{equation}
  \label{eq:ladderdef}
  [A,T] = \lambda A
\end{equation}
for some constant $\lambda$.  As a consequence of the definition, $A$
acts on eigenfunctions of $T$ by a spectral shift $\lambda$, possibly
annihilating finitely many eigenfunctions.  More generally, we say
that $A$ intertwines $T_1, T_2$ if $AT_1=T_2A$.  Thus,
\eqref{eq:ladderdef} is a special case of an intertwining relation
with $T_1=T,\; T_2=T+\lambda$.

In \cite{GGMM20} it was shown that, within the class of rational
extensions, the basic intertwiner between
$T_{M_1},T_{M_2},\; M_1,M_2\in \cM$ takes the form
\begin{equation}
  \label{eq:AMKdef}
  A_{M_1,K}[y]=\frac{\Wr[\psi_{M,k_1},\dots,\psi_{M,k_p},y]}{\Wr[\psi_{M,k_1},\dots,\psi_{M,k_p}]}, 
\end{equation}
where $K= M_1 \ominus M_2=\{k_1,\ldots, k_p\}$ is the index set of the
corresponding multi-flip $\flip_K$ that connects $M_1 \to M_2$, and
where the $\psi_{M,m},\; m\in \Z$ are the quasi-rational
eigenfunctions of $T_M$ defined in \eqref{eigenstates}.  One can show
that $A_{M,K}$ is a monic differential operator of order $p$ such that
\[ A_{M_1,K} T_{M_1} = T_{M_2} A_{M_1,K}.\] Operators $T_{M}$ and
intertwiners $A_{M,K}$ have the abstract structure of a category \cite{GGMM20}
because intertwiners $A_{M_1,K_1}$ and $A_{M_2,K_2}$ where
$K_1 = M_2\ominus M_1$  obey the following composition relation
\begin{equation}
\label{comp}
A_{M_2,K_2}\circ A_{M_1,K_1}=A_{M_1,K_{1}\ominus K_2}\circ p_{K_1,K_2}(T_M),
\end{equation}
where
\[
  \begin{aligned}
    p_{K_1,K_2}(m) &= \prod_{k\in K_1\cap K_2} (2k+1-m).
\end{aligned}
\]
Since $T_{M+n}=T_M+2n$, the above intertwiners
are also translation-invariant:
\begin{equation}
\label{translation}
A_{M+n,K+n}=A_{M,K},\quad n\in\Z.
\end{equation}
This allows us to consider a quotient category whose objects are
rational extensions modulo spectral shifts, and where the ladder
operators are precisely the endomorphisms.  For the details, see
Section 4 of \cite{GGMM20}.

For a Maya diagram $M\in \cM$ and an integer $n\in \Z$,  let
\begin{equation}
  \label{eq:LMndef}
  \LMn := A_{M,(M+n) \ominus M}.
\end{equation}
By Theorem 4.1 of \cite{GGMM20},
\begin{equation}
\label{ladder}
\LMn T_M=(T_M+2n)\LMn.
\end{equation}
Thus, $\LMn$ is a \textit{ladder operator} for the rational extension
$T_M$.  The action of $\LMn$ is that of a lowering or raising operator
according to
\[ L_n[\psi_{M,k}] = C_{M,n,k} \psi_{M,k-n},\quad k\notin M, \] where
$C_{M,n,k}$ is zero if $\psi_{M,k-n}$ is not a bound state, i.e., if
$k-n\in M$. Otherwise, $C_{M,n,k}$ is a rational number whose explicit
form is given in \cite{GGMM20}.  In general, the ladder operators
$L_{M,q},\; q\in \Z$ do not commute.  However, as we now show, there
is a natural subalgebra generated by lowering operators of certain
critical degrees $q$ that does commute.

\subsection{The annihilator algebra}
We say that a $q$th order ladder operator is an annihilator, if its
kernel is spanned by $q$ bound states.  The annihilator algebra of a
rational extension is more complicated than in the canonical case,
where the annihilator algebra is generated $L_-=\partial_x+x$. For a
non-empty partition, the analogous operators generate a non-trivial
algebra of commuting operators with a structure determined by the
combinatorics of the corresponding Maya diagram, as we now show.

For $q\in \N$, we say that a Maya diagram $M\in \cM$ is a $q$-core if
$M\subset M+q$.  We say that $q\in \N$ is a critical degree of a Maya
diagram $M\in \cM$ if $M$ is a $q$-core.  Observe that if $q$ is a
critical degree of $M$, then $q$ is a critical degree of $M+n$ for
every $n\in \N$.  Thus, the $q$-core property is an attribute of an
unlabelled Maya diagram.  The set of unlabelled Maya diagrams is
naturally bijective to the set of partitions, and so we use $\Dlam$,
where $\lambda$ is the partition corresponding to $M$, to denote the
set of all critical degrees. This definition is consistent with the
definition of the $q$-core partition used in combinatorics; see
\cite{Mc98} for more details.

A $q\in \N$ fails to be in $\Dlam$ if and only if there exists an
$m\in M$ and a $k\in I_M$ such that $q=m-k$.  The smallest empty
position on a Maya diagram occurs at position
$m_{\ell+1}+1 = \sigma_M-\ell$, while the largest occupied position
occurs at $m_1= \lambda_1-1+\sigma_M$.  It then follows that
\begin{equation}
  \label{eq:qcdef}
  q_c := m_1 - (\sigma_M-\ell)+1 =  \lambda_1+\ell
\end{equation}
is a threshold critical degree, in the sense that $q\in \Dlam$ for all
$ q\ge q_c$ and $q_c-1\notin \Dlam$.  See Figure
\ref{fig:lowering} for an example.

Let $K_q = (M+q)\ominus M,\; q\in \Z$ so that, by \eqref{eq:LMndef}
and \eqref{eq:AMKdef}, the kernel of $\LM{q}$ is spanned by
$\psi_{M,k},\; k\in K_q$.  By \eqref{eigenstates},
$L_{M,q},\; q\in \Z$ is an annihilator if and only if
$K_q\subset I_M$, if and only if $K_q = (M+q)\setminus M$, and if and
only if $M\subset M+q$.  Therefore, $L_{M,q}$ is an annihilator if and
only if $q\in \Dlam$ is a critical degree.

By Theorem 6.1 of \cite{KM20}, for every
critical degree $q\in \Dlam$, we have
\begin{equation}
  \label{eq:LqPsiz}
  \LM{q}(x,\partial_x) \Psilam(x,z) = z^q \Psilam(x,z).
\end{equation}
In other words, the generating function \eqref{eq:Psilamgf} is a joint
eigenfunction of the annihilators.  Let
$\Rlam =\lspan \{ z^q \colon q\in \Dlam\}$, and observe that if
$q_1, q_2\in \Dlam$, then $q_1+q_2\in \Dlam$ also.  It follows that
$\Rlam$ is closed with respect to multiplication; i.e$.$ $\Rlam$ is a
commutative algebra.  Also note that composition of annihilation
operators on the left of \eqref{eq:LqPsiz} is equivalent to
multiplication of eigenvalues on the right. It follows that the
annihilators commute, and that $\Rlam$ is isomorphic to the
annihilator algebra associated with the rational extension $T_M$.

Relations \eqref{eq:Psilamgf} and \eqref{eq:LqPsiz}
entail the following action of the annihilators on the bound states:
\begin{equation}
  \label{eq:Lqpsim}
  \LM{q}(x,\partial_x) \psi_{M,m}(x) = 2^q
  \gamma_{M,q}(m)\psi_{M,m-q}(x),
\end{equation}
where $m\in I_M, q\in \Dlam$, and where
\[ \gamma_{M,q}(m) = \prod_{k\in K_q} (m-k).\] Note that
$\gamma_{M,q}(m)=0$ when $\psi_{M,m}$ is a bound state, but
$\psi_{M,m-q}$ is not.

\subsection{Definition of the extended coherent states}
We now construct the ECS corresponding to a rational extension
$T_M,\; M\in \cM$.  We proceed, as in the canonical case, by
constructing the coherent state in terms of the generating
function. In \cite{KM20}, it was shown that, in terms of the
generating function $\Psilam(x,z)$, the eigenvalue relation
\eqref{eq:TMpsiMn} is equivalent to
\begin{equation}
  \label{eq:TMpsilam}
  T_M(x,\partial_x) \Psilam(x,z) = (z\partial_z + 1 + 2\sigma_M) \Psilam(x,z).
\end{equation}
Using the same change of variables as in \eqref{eq:Phi0def}, let us
therefore set
\begin{equation}
  \label{eq:Philamdef}
 \Philam(x,t;\alpha) = e^{-(1+2\sigma_M)it}\Psilam(x,\alpha e^{-2 it }).
\end{equation}
Then by construction, $\Philam(x,t)$ is an exact solution of the
time-dependent Schr\"odinger equation corresponding to the rational
extension $T_M$:
\begin{equation}
  \label{eq:PhilamTM}
    i \partial_t \Philam(x,t) = T_M(x,\partial_x) \Philam(x,t) .
\end{equation}
Applying the same change of variables to \eqref{eq:LqPsiz}, we obtain
the annihilator eigenrelation
\begin{equation}
  \label{eq:LqPhial}
  \LM{q}(x,\partial_x) \Philam(x,t;\alpha) = \alpha^q e^{-2iqt} \Philam(x,z).
\end{equation}
Hence, $\Philam(x,t;\alpha)$ is a joint eigenfunction of the
annihilator algebra and satisfies the definition of a
coherent state in the sense of Barut-Girardello \cite{BG71}.

\subsection{Minimized uncertainty}
The canonical coherent state $\Phi_0(x,t;\alpha)$ saturates the
Heisenberg uncertainty relation for position and momentum.  Formally,
we have
\[ E(\Delta x)^2 E(\Delta p)^2 = \frac14 \]
where
\begin{equation}
  \label{eq:canonvars}
\begin{aligned}
  E(\Delta x)^2 &= \frac{\int_\R  x^2
                 \Phi_0\overline{\Phi_0}dx}{
                 \int_\R \Phi_0   \overline{\Phi_0} dx}  -
  \lp\frac{\int_\R x\Phi_0 \overline{\Phi_0}  dx}{\int_\R \Phi_0
  \overline{\Phi_0} dx}\rp^2 \\
  E(\Delta p)^2 &= -\frac{\int_\R (\partial_{xx}
                 \Phi_0)\overline{\Phi_0}dx}{
                 \int_\R \Phi_0   \overline{\Phi_0} dx}  - 
                 \lp\frac{\int_\R i(\partial_x\Phi_0) \overline{\Phi_0}
                 dx}{
                 \int_\R \Phi_0 \overline{\Phi_0} dx}\rp^2
\end{aligned}
\end{equation}

Without loss of generality, $M=\Mlam$, whence by \eqref{eq:Psilamdef} and the definition
\eqref{eq:Philamdef} we see that $\Philam \to \Phi_0$ as
$\alpha\to +\infty$. Consequently for an ECS, the minimized
uncertainty relation holds asymptotically, in the sense that
\begin{equation}
  \label{eq:minunlam}
   E_\lambda(\Delta x)^2 E_\lambda(\Delta p)^2 \to \frac14,\quad
   \text{as }
   \alpha \to + \infty, 
\end{equation}
where $E_\lambda(\Delta x), E_\lambda(\Delta p)$ denote the
expectation values of the variances associated with the wave function
$\Philam(x,t;\alpha)$. Formally, these are defined in the same way as
\eqref{eq:canonvars}, but with $\Philam$ in place of $\Phi_0$.

\subsection{Example} As an example, we construct the coherent state
corresponding to the index set $K=\{2,3\}$.  The corresponding Maya diagram, partition, and index are
\[ M= f_K(\Mz) = \{ \ldots, -2,-1, 2,3 \},\quad \lambda=(2,2),\quad
  \sigma_M=2,\]
while the corresponding rational extension is
\[ T_M(x,\partial_x) =-\partial_x^2+ \left(x^2+4+\frac{32 x^2}{4
      x^4+3}-\frac{384 x^2}{\left(4 x^4+3\right)^2}\right). \]
The bound states are indexed by
\[ I_M = \Z\setminus M = \{ 0,1,4,5,6,\ldots \} , \] and the bound state
with eigenvalue $2m+1,\; m\in I_M$ given by
\[ \psi_{M,m} =  e^{-\tfrac12x^2} \frac{\hH_{M,m}(x)}{4x^4+3}
,\quad m\in I_M,\] where the corresponding exceptional polynomial is
\[ \hH_{M,m} = \frac{\Wr(2x^2-1,2x^3-3,H_m)}{4 (m-2)(m-3)},\quad m\ge
  0,\; m\neq 2,3. \]
%Table \ref{tab:HMm} shows the first few of these exceptional polynomials.
% \begin{table}
%   \centering
%   \begin{tabular}{r|l}
%     $m$ & $H_{M,m}(x)$\\ \hline
%     0 & $x^2+\frac{1}{2}$\\
%     1 & $2 x^3+3 x$\\
%     4 & $16 x^6+24 x^4+36 x^2-18$\\
%     5 & $32 x^7+16 x^5+40 x^3-60 x$\\
%     6 & $64   x^8-64 x^6-240 x^2+60$
%   \end{tabular}
%   \caption{Bound states of the rational extension corresponding to the
%     index set $K=\{ 2,3\}$.}
%   \label{tab:HMm}
% \end{table}

In this case, the Schur function is
\[ \Slam(t_1,t_2,t_3) = \frac{t_1^4}{12}+t_2^2-t_1 t_3.\] Using
\eqref{eq:Psilamdef}, the generating function for the bound states is
therefore
\begin{equation}
  \label{eq:psilamex}
 \Psilam(x,z) = \lp 1- \frac{16x^3}{4x^4+3} z^{-1}+
  \frac{12(2x^2+1)}{4x^4+3} z^{-2} \rp
  e^{-\tfrac12(x-z)^2+\tfrac14z^2}.
\end{equation}

The set of critical degrees is $\Dlam = \{ 4,5,\ldots \}$. Note that there are no critical degrees below the threshold $q_c=2+2=4$.
Figure \ref{fig:lowering} illustrates the fact that $q=4$ is a
critical degree and that $q=3$ fails to be a critical degree since
$0+3\in M$ but $0\notin M$.

The extended coherent state  
\[ \Philam(x,t;\alpha) = e^{-5 it} \Psilam(x,\alpha e^{-2it}) \] is an
exact solution of the corresponding time-dependent Schr\"odinger
equation \eqref{eq:PhilamTM}.  The first 4 annihilators, as defined in
\eqref{eq:AMKdef}, are $\LM{q}=A_{M,K_q},\; q\in \{4,5,6,7\}$ with
\[ K_4 =\{0,1,6,7\},\; K_5=\{0,1,4,7,8\},\; K_6=\{0,1,4,5,8,9\},\;
  K_7=\{0,1,4,5,6,9,10\}. \]
% \begin{align*}
%   \LM4  &= A_{M,\{0,1,6,7\}},\\
%   \LM5  &= A_{M,\{0,1,4,7,8\}},\\
%   \LM6  &= A_{M,\{0,1,4,5,8,9\}},\\
%   \LM7  &= A_{M,\{0,1,4,5,6,9,10\}}
% \end{align*}
These commuting differential operators generate the annihilator
algebra of this rational extension.  In each case, one can verify by
direct calculation that $\Philam(x,t;\alpha)$ is an eigenfunction of
$\LM{q}$ with eigenvalue $\alpha^q e^{-2qi t}$.

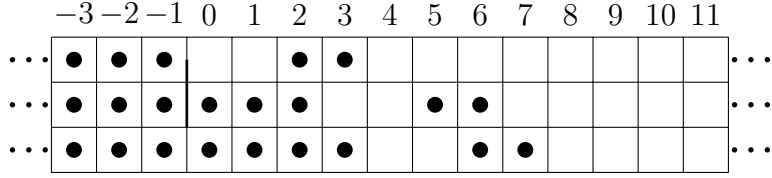
\begin{figure}[h]
	\centering
	\begin{tikzpicture}[scale=0.6]
	
	\path [fill] (0.5,2.5) node {\huge ...}
	++(1,0) circle (5pt) ++(1,0) circle (5pt)  ++(1,0) circle (5pt)
	++(3,0) circle (5pt) ++(1,0) circle (5pt)
	++(9,0) node {\huge ...} +(1,0);

	\foreach \x in {-3,...,11} 	\draw (\x+4.5,3.5)  node {$\x$};
	\path (6.5,0.5) node {} ++ (1,0) node {}
	++ (2,0) node {}++ (2,0) node {}++ (3,0) node {}
	;
	
	\draw  (1,0) grid +(15 ,3);
	
	\path [fill] (0.5,1.5) node {\huge ...}
	++(1,0) circle (5pt)
        ++(1,0) circle (5pt)
	++(1,0) circle (5pt)
	++(1,0) circle (5pt)
	++(1,0) circle (5pt)
	++(1,0) circle (5pt)
	++(3,0) circle (5pt) ++(1,0) circle (5pt)
	++(6,0) node {\huge ...} ;
	\path [fill] (0.5,0.5) node {\huge ...}
	++(1,0) circle (5pt) ++(1,0) circle (5pt)  ++(1,0) circle (5pt)
	++(1,0) circle (5pt)
	++(1,0) circle (5pt)
	++(1,0) circle (5pt)
	++(1,0) circle (5pt)
	++(3,0) circle (5pt) ++(1,0) circle (5pt)
	++(5,0) node {\huge ...} +(1,0);
	%node[anchor=west] { $M =  (-\infty,b_0)\cup [ b_1,b_2) \cup [ b_3,b_4)$};
	
	\draw[line width=1pt] (4,1) -- ++ (0,1.5);
	
	\end{tikzpicture}
%	\captionsetup{font=footnotesize}
	\caption{Top: The Maya diagram $M$ corresponding to index set
          $K=\{2,3\}$.  The corresponding partition and index are
          $\lambda=(2,2)$ and $\sigma_M=2$, respectively, while the threshold
          critical degree is $q_c = 4$. Middle: $M+3$. Bottom:
          $M+4$. Note that $4$ is a critical degree since
          $M\subset M+4$. However, $3$ fails to be a critical degree
          since $3\in M$ but $3\notin M+3$. }
        \label{fig:lowering}
\end{figure}

The form of \eqref{eq:psilamex} makes evident the asymptotic relation
$\Psilam \to \Psi_0$ as $\alpha=|z|\to +\infty$.  Figure
\ref{fig:uncert} shows the value of the position-momentum uncertainty
value $E_\lambda(\Delta x)^2 E_\lambda(\Delta p)^2$ as a function of
time $t$ for values $\alpha=4,8,16$.  The graphs clearly indicate the
corresponding asymptotic minimization of the uncertainty relation as
$\alpha\to +\infty$.

\begin{figure}
  \centering
  \includegraphics[width=8cm]{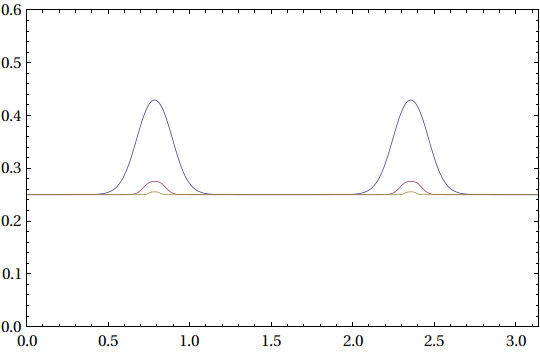}
  \caption{The position-momentum uncertainty value
    $E_\lambda(\Delta x)^2 E_\lambda(\Delta p)^2$ as a function of
    time for the ECS with $\lambda=(2,2)$ and $\alpha=4$ (blue), $8$ (red), and $16$ (yellow).}
  \label{fig:uncert}
\end{figure}


\begin{thebibliography}{99}
\bibitem{BG71} Barut, A.O., Girardello,  L.: New coherent states
  associated with non compact groups. Commun Math Phys.
  \textbf{21}  41-55 (1971)
\bibitem{PE86} Perelomov, A.: Generalized coherent
  states and their applications. Springer-Verlag, Heidelberg  (1986)
\bibitem{GL63} R.J. Glauber. The quantum theory of optical
  coherence. \textit{Phys Rev.} \textbf{130} (1963) 2529.
\bibitem{FHRO07} Fern\'andez, D.J., Hussin, V., Rosas-Ortiz, O.: Coherent
  states for Hamiltonians generated by supersymmetry. J. Phys. A
  \textbf{40}, 6491 (2007)
\bibitem{BGBM99} Bagchi, B., Ganguly, A., Bhaumik, D., Mitra, A.:
  Higher derivatives supersymmetry, a modified Crum-Darboux
  transformation and coherent state. Mod.  Phys. Lett. A \textbf{14}, 27–34
  (1999)
\bibitem{GDH21}
Garneau-Desroches,   S., Hussin,  V.: Ladder operators and coherent
  states for the Rosen-Morse system and its rational
  extensions. J. Phys. A Math. Theor.  \textbf{54}, 475201 (2021)
\bibitem{GGM13}  {David G\'omez-Ullate, Yves Grandati, and Robert
    Milson, Rational extensions of the quantum harmonic oscillator and
    exceptional Hermite polynomials, \textit{J. Phys. A} \textbf{47}
    (2013), no. 1, 015203.}
\bibitem{HHMYZ18} Hoffmann, S.E., Hussin, V., Marquette, I.,
  Yao-Zhong, Z.: Non-classical behaviour of coherent states for
  systems constructed using exceptional orthogonal
  polynomials. J. Phys. A Math. Theor. \textbf{51} 085202 (2018)
\bibitem{MQ13} Marquette I, Quesne C.: New ladder operators for a
  rational extension of the harmonic oscillator and superintegrability
  of some two-dimensional systems. J. Math.
  Phys. \textbf{54} 102102 (2013)
\bibitem{CAFMC22} Contreras-Astorga A., Fern\'andez C., Muro-Cabral
  C. Equivalent non-rational extensions of the harmonic oscillator,
  their ladder operators and coherent states. Europea
  Physi. J. Plus.\textbf{138} (2023)
  systems. Advances in Mathematics. \textbf{385} 107770 (2021)
\bibitem{CP17} Cari\~nena JF, Plyushchay MS: ABC of ladder operators
  for rationally extended quantum harmonic oscillator
  systems. J. Phys. A,\textbf{50} 275202 (2017)
\bibitem{BCAF15} Bermudez D, Contreras-Astorga A, Fernández DJ:
  Painlev\'e IV Hamiltonian systems and coherent states. J Phys: Conf
  Ser. \textbf{597} 012017 (2015)
\bibitem{GGM21} G\'omez-Ullate D, Grandati Y, Milson R.: Complete
  classification of rational solutions of A2n-Painlev\'e
\bibitem{KM20} A. Kasman and R. Milson, The Adelic Grassmannian and
  Exceptional Hermite Polynomials, \textit{Math. Phys. Analysis and
    Geometry} \textbf{23} (2020), 1-51.
\bibitem{GGMM20} D.  G\'omez-Ullate, Y. Grandati, Z. McIntyre and
  R. Milson, Ladder operators and rational extensions,
  \textit{Proceedings of QTS 11}, CRM Series on Math. Phys.,
  2020.
\bibitem{Sch26} Schr\"odinger, E.: Der stetige Übergang von der
  mikro-zur makromechanik. Naturwissenschaften \textbf{14} 664–666
  (1926)
\bibitem{GGM18} {D. G\'omez-Ullate, Y. Grandati, and R. Milson. Durfee
    rectangles and pseudo-Wronskian equivalences. \textit{Studies in
      Applied Mathematics} \textbf{141} (2018): 596-625.}
\bibitem{HHV16} W.A. Haese-Hill, M.A. Hallnas, and A.P. Veselov,
  Complex Exceptional Orthogonal Polynomials and Quasi-invariance,
  \textit{Lett. Math. Phys.}, \textbf{106} (2016) 583-606.
\bibitem{KS85}
  J.R. Klauder and B. Skagerstam. \textit{Coherent states:
    applications in physics and mathematical physics}. World
  scientific, 1985.
\bibitem{Mc98} Macdonald IG. \textit{Symmetric functions and Hall
    polynomials}. Oxford press, 1998.

\bibitem{NY99} M. Noumi and Y. Yamada, Symmetries in the fourth Painlev\'e
  equation and Okamoto polynomials, \textit{Nagoya Math. J.}
\textbf{153} (1999), 53-86.
\end{thebibliography}
\end{document}